%% file: paper_draft.tex
\documentclass[conference]{IEEEtran}
\usepackage{amsmath,amsthm,amssymb}
\usepackage{graphicx}
\usepackage{float}
\usepackage{xcolor}
\usepackage{tikz}
\usetikzlibrary{patterns}
\newcommand{\xtic}[1]{(#1,0) -- ++(0,-0.25)}
\newcommand{\ytic}[1]{(0,#1) -- ++(-0.25,0)}
\newcommand{\I}[1]{I_{\set{#1}}}
\usepackage{cite}
\usepackage{hyperref}

\input{macros}      
\usepackage{texdef2020}
\newcommand{\spmf}[2]{\pmf{}{#2}}
\newcommand{\support}[2][n]{\mathcal{#2}_{#1}}

\title{Privacy Leakage in Discrete-Time Updating Systems}
\author{Nitya Sathyavageeswaran, Roy D. Yates, Anand D. Sarwate, Narayan Mandayam \\
Department of Electrical and Computer Engineering, Rutgers University\\ E-mail: \{nitya.s, anand.sarwate\}@rutgers.edu, \{ryates, narayan\}@winlab.rutgers.edu}
\usepackage{url}
\begin{document}

\maketitle

\begin{abstract}
A source generates time-stamped update packets that are sent to a server and then forwarded to a monitor. This occurs in the presence of an  adversary that can infer information about the source by observing the output process of the server. The server wishes to release updates in a timely way to the monitor but also wishes to minimize the information leaked to the adversary.  
We analyze the trade-off between the age of information (AoI) and the maximal leakage for systems in which the source generates updates as a Bernoulli process. For a time slotted system in which sending an update requires one slot, we consider three server policies:
(1) Memoryless with Bernoulli Thinning (MBT): arriving updates are queued with some probability  and head-of-line update is released after a geometric holding time; (2) Deterministic Accumulate-and-Dump (DAD):  the most recently generated update (if any) is released after a fixed time; (3) Random Accumulate-and-Dump (RAD):  the most recently generated update (if any) is released after a geometric waiting time. We show that for the same maximal leakage rate, the DAD policy achieves lower age compared to the other two policies but is restricted to discrete age-leakage operating points.   
\end{abstract}


\section{Introduction}

A smart home has a range of devices that can be accessed or controlled remotely using the internet. The various sensors in a smart home send time-stamped updates about the temperature, humidity, power consumption, etc.~to a destination monitor.
%
%
In spite of the advantages offered by a smart home system, there is also a potential loss of privacy. 
Adversaries could infer 
sensitive information about the home occupants from the update packets generated from the various sensors.
Moreover, timeliness of the updates is important in many applications, but  improving the timeliness of delivered updates can increase the information leaked to the adversary. 
When the timeliness of these updates is important, an age of information metric is useful in optimizing these systems. The age of information metric tells us how much time has elapsed  since the generation of the most recent update that has been received at the monitor~\cite{kaul2012real}.


In this paper we study trade-offs between privacy and age. We consider a model for a smart home system in which a source generates time-stamped updates that are sent to a monitor through a server. There is an adversary present at the  monitor that can infer information about the source from the packet arrival process. We measure privacy using the maximal leakage metric introduced by Issa, Kamath, and Wagner~\cite{7460507}; this measures the maximal multiplicative gain that the adversary can guess any function of the original data from the observed data.  The maximal leakage increases when we minimize  age in these systems.  We design service policies in order to reduce the maximal leakage and see the effect on the age of information (AoI) for these systems. 

The time average AoI  for various systems has been extensively studied, including the  first-come first-served (FCFS) M/M/1, M/D/1 and D/M/1 queues, and the last-come first-served (LCFS) queues~\cite{kaul2012real, 7364263, 6875100, kaul2012status,najm2016age, bedewy2016optimizing,  yates2018age, najm2018status}. AoI  has also been analyzed for other continuous time queueing disciplines~\cite{ Yates2020TheAO,  8006592, soysal2019age, 8437907, talak2019heavy, kam2016controlling, kam2018age, inoue2018analysis} and for various discrete time queues~\cite{9148775, tripathi2019age, article, talak2019optimizing}. 
Minimizing the timing information leaked, while keeping the status updates timely in an energy harvesting channel is analyzed in~\cite{ozel2022state}. 
%


Minimizing the maximal leakage subject to a cost constraint has been studied in~\cite{9162276, 8613385, 8006634, 7541353}. Issa et al.~\cite{8943950} derive the maximal leakage for both an M/M/1 queue and an ``accumulate-and-dump'' system. 

Our contribution is to study the trade-off between maximal leakage and the age. We examine the discrete time analogues of the M/M/1 queue and the ``accumulate-and-dump'' system.  We also introduce a ``Random Accumulate-and-Dump'' service policy. We derive the age and maximal leakage for these systems. 
This raises interesting questions about the general problem of balancing timeliness and privacy in communication systems.  


\section{System Model and Paper Overview}

Our 
discrete time  system consists of a source, server, monitor and an adversary as shown in Figure \ref{fig:model}. Time passes in integer slots with slot $n\ge0$ denoting the time interval $[n,n+1)$. The transmission of a packet requires one slot. A packet sent in slot $n$ (from source to server or from server to monitor) is received at the start of slot $n+1$.  

\subsection{Information Leakage}
To characterize information leakage, packet transmissions from the source to the server are indicated  by the binary sequence $X^n=(X_1,\ldots, X_n)$ such that $X_t=1$ if the server receives a packet from the source at the start of slot $t$. This packet will have been generated and transmitted by the source in slot $t-1$. Packet transmissions from  the server to the monitor are similarly indicated by the binary sequence $Y^n=(Y_1,\ldots, Y_n)$ such that $Y_t=1$ if the server sends a packet in slot $t$. The adversary observes the server transmission sequence $Y^n$. The updating policy of the source, coupled with the updating policy of the server induces a joint PMF $\pmf{X^n,Y^n}{x^n,y^n}$. From this PMF, the support set of $X^n$ is denoted 
$\support{X}=\set{x^n\colon \pmf{X^n}{x^n}>0}$.
The support set $\support{Y}$ of $Y^n$ is defined analogously.  We measure the information leaked to the adversary by the maximal leakage metric.
\begin{definition}[Issa et al.~\cite{8943950}]
A joint distribution $P_{X^nY^n}$ on finite alphabets $\mathcal{X}^n$ and $\mathcal{Y}^n$ has  maximal leakage 
\begin{align}
\L(X^n\to Y^n)=\log \sum_{y^n\in \support{Y}} \max_{x^n\in\support{X}} \pmf{Y^n|X^n}{y^n|x^n}. \eqnlabel{maxL_def}
\end{align}
\end{definition}
The key technical challenge in this work is to find the leakage maximizing input sequence $x^n\in \support{X}$ for each possible output sequence $Y^n=y^n$. Note that the maximal leakage depends on the arrival process $X^n$ only through its support $\support{X}$.  We find the leakage maximizing input $x^n$ for {\em full support} arrival processes satisfying 
$\support{X}=\set{0,1}^n$.
We further assume the source sends fresh updates as a rate $\lambda$ Bernoulli process since this is the simplest class of arrival processes with full support.

%


\subsection{Age of Information}
For the evaluation of update timeliness, we employ a discrete-time age process model that is consistent with prior work~\cite{Kosta2019QueueMF}. The source generates time-stamped update packets (or simply updates) that are sent to the server.  An update generated by the source in slot $t$ is based on a measurement of a process of interest that is taken at the beginning of the slot and has time-stamp $t$.  At the end of slot $t$, or equivalently at the start of slot $t+1$, that packet will have age $1$. In slot $t+j$, this update will have age $j$. We say one packet is fresher than another if its age is smaller. 

An observer of these updates measures timeliness by a discrete-time {\em age} process $A(t)$ 
that is constant over a slot and equals the age of the freshest packet received prior to slot $t$. When $u(t)$ denotes the time-stamp of the freshest packet observed/received prior to slot $t$, the age in slot $t$ is $A(t)=t-u(t)$. We characterize the timeliness performance of the system by the average age of information (AoI) at the monitor.  
\begin{definition}[Kaul et al.~\cite{kaul2012real}]
A stationary ergodic age process has age of information (AoI) 
\begin{align}
   \E{A(t)}=\limty{T} \frac{1}{T}\sum_{t=0}^{T-1} A(t)\eqnlabel{sum-age},
\end{align} 
where $\E{\cdot}$ is the expectation operator.
\end{definition}

%

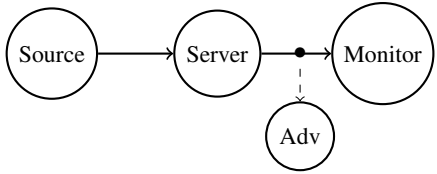
\begin{figure}[t]
\centering
\begin{tikzpicture}[node distance=1.1cm]
\node[draw,circle,thick] (source) {\small Source};
\node (mid1)[right of=source] {};
\node[draw,circle,thick] (server)[right of=mid1] {\small Server};
\node (mid2)[right of=server] {$\bullet$};
\node[draw,circle,thick] (mon)[right of=mid2] {\small Monitor};
\node[draw,circle,thick] (oppo)[below of=mid2] {\small Adv};
\draw[->, thick]  (source.east) -- (server.west);
\draw[->, thick]  (server.east) -- (mon.west);
\draw[->,dashed] (mid2) -- (oppo);
\end{tikzpicture}
\caption{A source sends status updates through a server to a monitor. An adversary (Adv) observes transmissions from the server to the monitor.}
\label{fig:model}
\end{figure}

\subsection{Server Policies}
\label{sec:serverpolicies}
 The server wishes to balance the timing information leaked to the adversary against the AoI at the monitor. To characterize these conflicting objectives, we explore age-leakage trade-offs for three server policies:  



\begin{itemize}
\item \textbf{Memoryless with Bernoulli thinning (MBT)}: The server admits each arriving update into a first-come first-served queue with probability $\alpha$. 
If the queue is not empty at the start of a slot, the update at the head of the queue is sent to the monitor with probability $\mu$, independent of transmissions in other slots. 
    
\item \textbf{Deterministic Accumulate-and-Dump (DAD)}: The server stores the freshest update received  from the source.
    For a fixed parameter $\tau$, the server sends its stored update every $\tau$ slots.
    
\item \textbf{Random Accumulate-and-Dump (RAD)}: The server stores the freshest update received from the source.  In each slot $t$, the server sends its stored update with probability $\mu$, independent of transmissions in other slots. 
\end{itemize}

With rate $\lambda$ Bernoulli arrivals, the MBT server acts as a Geo/Geo/1 queue with effective arrival rate $\alpha\lambda$. Furthermore, while the service time of an update is always one slot, each update spends a geometric number of slots at the head-of-line (HOL) position of the queue and the departure process is indistinguishable from that of a discrete-time Geo/Geo/1 queue. It is the discrete-time version of the M/M/1 queue.  

The 
DAD policy 
was introduced Issa et al.~\cite{8943950} with the generalization that up to $n$ packets are dumped at a time.  Here, we focus on the special case of DAD in which no more than $n=1$ packet  is ``accumulated'' and dumped because maximal leakage grows linearly with $n$~\cite{8943950} and because the AoI is the same for all $n\ge1$ as long as the dumped updates include the most recent update. For the same reasons, the RAD server also accumulates only one packet.  We  further note that with Bernoulli arrivals, RAD  is the discrete-time analogue of the M/M/1/1 queue with preemption in service; the time an update spends in the HOL position is geometric and in each time slot the HOL update may be replaced by a new arrival. 

\section{Maximal Leakage}

A fixed arrival sequence $x^n$ and service policy induces a conditional PMF $\pmf{Y^n|X^n}{y^n|x^n}$. For the MBT and RAD servers, the service time of a packet is geometric with parameter $\mu$. 
For these servers, we employ the 
following lemma 
to calculate the maximal leakage.
\begin{lemma}
\label{max_claim}
For the MBT 
and RAD 
servers with full support arrival processes,
\begin{align}
    \max_{x\in {\mathcal{X}}^n} \pmf{Y^n|X^n}{y^n|x^n}=\mu^{\sum_{i=1}^{n} y_i}
\end{align} and this maximum is achieved when $x^n=y^n$.
\end{lemma}
The proof follows by an induction argument; details are in the extended version of this manuscript.
Lemma~\ref{max_claim} shows that for a given output sequence $y^n$ from either the MBT or the RAD server, the maximal leakage is maximized by a ``just-in-time'' input $x^n=y^n$ in which the source generates each update in the slot just prior to its departure from the server. The MBT and the RAD server share the property that  the HOL update in each slot is transmitted with probability $\mu$. The just-in-time input $x^n=y^n$ maximizes  $\pmf{Y^n|X^n}{y^n|x^n}$ because it precludes departures from the server prior to the departures specified by the output $y^n$. 

\begin{theorem}
\label{Theorem1}
For full support arrival processes, the maximal leakage rate for the MBT and RAD servers is given by
\begin{align}
    \frac{1}{n} \L(X^n\to Y^n)=\log (1+\mu).
\end{align}
\end{theorem}

\begin{proof}
Applying Lemma \ref{max_claim} to \eqnref{maxL_def} yields
 \begin{align}
    \L(X^n\to Y^n)&= \log \sum_{y^n\in \Y^n} \mu^{\sum_{i=1}^{n}y_i}\nn
    &=\log \sum_{y_1=0}^{1}\mu^{y_1}\cdots \sum_{y_n=0}^{1}\mu^{y_n}\nn
    &=\log {(\mu^0+\mu^1)}^n =n\log (1+\mu).
\end{align}
\end{proof}








For the DAD policy, the most recently generated packet is dumped after every $\tau$ slots and $Y^n$ is  a deterministic function of $X^n$.  
Defining $K=\floor{n/\tau}$,
 $Y^n$ has the form 
 \begin{align}\eqnlabel{Yn-DADsupport}
Y^n=(0^{\tau-1},Y_\tau,0^{\tau-1},Y_{2\tau},\ldots,Y_{K\tau},0^{n-K\tau})
\end{align}
where $Y_{k\tau}=0$ if and only if $(X_{(k-1)\tau+1},\cdots, \X_{k\tau})=0^\tau$  and otherwise  $Y_{k\tau}=1$. This structure simplifies the leakage calculation.
\begin{lemma}
\label{Claim_2}
For the DAD policy, $\max_{x\in {\mathcal{X}}^n} \pmf{Y^n|X^n}{y^n|x^n}=1$ and the maximum is achieved when $x^n=y^n$.
\end{lemma}
\begin{proof}
For a given output $Y^n=y^n$ from the DAD server,  the input $x^n=y^n$ implies $\pmf{Y^n|X^n}{y^n|x^n}=1$. This is the maximizing input since it achieves the unity upper bound. 
 \end{proof}
We note that the maximizing $x^n$ may not be unique.
\begin{theorem}\thmlabel{DAD-leakage}
The DAD
policy has maximal leakage rate 
\begin{align}
\frac{\L(X^n\to Y^n)}{n}
&=\frac{\log 2}{n}\floor{\frac{n}{\tau}}.
\end{align}
\end{theorem}

\begin{proof}
Since $Y^n$ has the form \eqnref{Yn-DADsupport}, the support set $\support{Y}$ has size $\abs{\support{Y}}=2^K=2^{\floor{n/\tau}}$.  Applying Lemma~\ref{Claim_2} to \eqnref{maxL_def}, 
 \begin{align}
    \L(X^n\to Y^n)&= \log \sum_{\mathclap{y^n\in \support{Y}}} 1
    =\log {2}^{\left \lfloor{n/\tau}\right \rfloor} \eqnlabel{dump_instance}
    =\floor{\frac{n}{\tau}} \log 2.
\end{align}
\end{proof}
Since  $\limty{n}\floor{n/\tau}/n=1/\tau$, it follows from \Thmref{DAD-leakage} that the 
DAD policy has maximal leakage rate
\begin{align}
    \limty{n}\frac{\L(X^n\to Y^n)}{n}&=\frac{\log 2}{\tau}.
\end{align}

\section{Age of Information}
The MBT server with rate $\lambda$ Bernoulli arrivals and admission probability $\alpha$ is identical to a Geo/Geo/1 queue with arrival rate $\alpha\lambda$. 
If the server queue is not empty at the start of slot $t$, the server sends the head-of-line update with probability $\mu$. Using the discrete time analogue of the M/M/1 graphical AoI analysis~\cite{kaul2012real}, the AoI of the Geo/Geo/1 queue has been previously derived by Kosta et al.~\cite{Kosta2019QueueMF}  and equals
\begin{align}\eqnlabel{GeoAge}
    \E{A_{\text{MBT}}}=\frac{1}{\alpha\lambda}+\frac{1-\alpha\lambda}{\mu-\alpha\lambda}-\frac{\alpha\lambda}{\mu^ 2}+\frac{\alpha\lambda}{\mu}.
\end{align}

While the graphical method could be employed for age analysis of the DAD and RAD systems, here we instead employ the sampling of age processes approach~\cite{Yates2020TheAO}. In the implementation of the DAD/RAD servers, no update is dumped if no update arrived in the preceding inter-dump period. However, for the purpose of AoI analysis, we make a different assumption: if no update arrives before the dump attempt, then the server resends its previously dumped update. This approach is an example of the method of fake updates introduced in Yates and Kaul~\cite{6284003}. With respect to the age at the monitor, in the absence of an update to dump, it makes no difference if the server sits idle  or if the server repeats sending the prior update. Neither changes the age at the monitor. 

However this fake update approach 
simplifies the age analysis.  We can now view the monitor as {\em sampling} the update process of the server. With each dump instance, the monitor receives the freshest update of the server,  resetting the age at monitor to the age of the dumped update. In the same way, we can view the server as  sampling the always-fresh update process of the source.

For both the DAD and RAD systems, we observe that the source offers fresh updates to the server as a renewal process. The dump attempts of the DAD and RAD servers are also renewal processes. Specifically, the DAD server attempts to release its freshest update every $\tau$ slots while the RAD server has dump attempts with iid geometric inter-dump times.

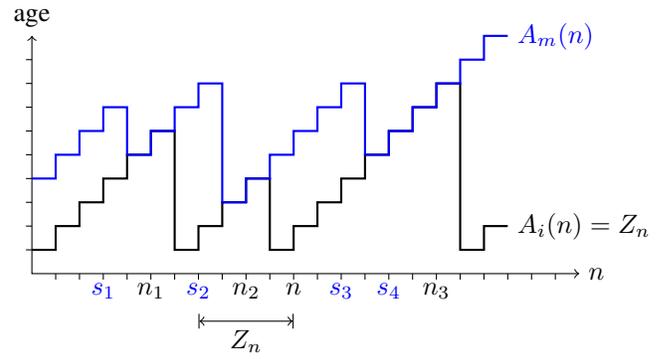
\begin{figure}[t]
    \centering
\begin{tikzpicture}[scale=\linewidth/28cm]
\draw [<->] (0,10) node [above] {age} -- (0,0) -- (23,0) node [right] {$n$};
\draw [thin]\xtic{1}\xtic{2}\xtic{3}\xtic{4}\xtic{5}\xtic{6}\xtic{7}
\xtic{8}\xtic{9}\xtic{10}\xtic{11}\xtic{12}\xtic{13}\xtic{14}\xtic{15}\xtic{16}\xtic{17}\xtic{18}\xtic{19}\xtic{20}\xtic{21}\xtic{22}; 
\draw [thin]\ytic{1}\ytic{2}\ytic{3}\ytic{4}\ytic{5}\ytic{6}\ytic{7}\ytic{8}\ytic{9};
\draw [thick] 
(0,1) -- ++(1,0) -- 
++(0,1) -- ++(1,0) --
++(0,1) -- ++(1,0) --
++(0,1) -- ++(1,0) --
++(0,1) -- ++(1,0) --
++(0,1) -- ++(1,0) --
++(0,-5)-- ++(1,0) -- 
++(0,1) -- ++(1,0) --
++(0,1) -- ++(1,0) --
++(0,1) -- ++(1,0) --
++(0,-3)-- ++(1,0) -- 
++(0,1) -- ++(1,0) --
++(0,1) -- ++(1,0) --
++(0,1) -- ++(1,0) --
++(0,1) -- ++(1,0) --
++(0,1) -- ++(1,0) --
++(0,1) -- ++(1,0) --
++(0,1) -- ++(1,0) --
++(0,-7) -- ++(1,0) -- 
++(0,1) -- ++(1,0) node [right] {$A_i(n)=Z_n$};
\draw [thick,blue] 
(0,4) -- ++(1,0) -- 
++(0,1) -- ++(1,0) --
++(0,1) -- ++(1,0) --
++(0,1) -- ++(1,0) --
++(0,-2)-- ++(1,0) -- 
++(0,1) -- ++(1,0) --
++(0,1) -- ++(1,0) --
++(0,1) -- ++(1,0) --
++(0,-5)-- ++(1,0) -- 
++(0,1) -- ++(1,0) --
++(0,1) -- ++(1,0) --
++(0,1) -- ++(1,0) --
++(0,1) -- ++(1,0) --
++(0,1) -- ++(1,0) --
++(0,-3)-- ++(1,0) -- 
++(0,1) -- ++(1,0) --
++(0,1) -- ++(1,0) --
++(0,1) -- ++(1,0) --
++(0,1) -- ++(1,0) --
++(0,1) -- ++(1,0) node [right] {$A_m(n)$};
\draw 
(5,0) node [below] {$n_1$} 
(9,0) node [below] {$n_2$} 
(11,0) node [below] {$n$} 
(17,0) node [below] {$n_3$};
\draw 
(3,0) node [below,blue] {$s_1$} 
(7,0) node [below,blue] {$s_2$} 
(13,0) node [below,blue] {$s_3$}
(15,0) node [below,blue] {$s_4$};
\draw [|<->|] (7,-2) to node [below] {$Z_n$} ++(4,0);
\end{tikzpicture}
\caption{The source sends fresh updates to the server in slots $N_k=n_k$, inducing the age process $A_i(n)$ at the input to the server. The server sends samples of the most recent update to the monitor in time slots $S_k=s_k$, inducing the age process $A_m(n)$ at the monitor.}
\label{fig:DTsampler}
\end{figure}

\subsection{Sampling the Source: The Age of Fresh Updates}
\label{sec:fresh}
Referring to Figure~\ref{fig:model}, the source can generate a fresh (age zero) update in a slot $n$ and forward it to the server in that same slot. This fresh update arrives at the server at the end of slot $n$ with age $1$.  The age process $A_i(n)=Z_n$ of an observer at the server input is reset at the start of slot $n+1$ to $Z_{n+1}=1$.  When a source generates fresh updates in slots $N_k=n_k$, the age $Z_n$ evolves as the sequence of staircases shown in Figure~\ref{fig:DTsampler}. 
From the observer's perspective, update $k$ {\em arrives} in slot $N_k+1$ with  interarrival time 
\begin{align}
Y_k=N_k+1- (N_{k-1}+1)=N_k-N_{k-1}.
\end{align}

Defining the indicator $\I{A}$ to be $1$ if event $A$ occurs and zero otherwise, we can employ Palm probabilities to calculate the age PMF 
\begin{align}
\prob{Z_n=z}=\limty{N}\frac{1}{N}\sum_{n=0}^{N-1} \I{Z_n=z}.\eqnlabel{PalmZ}
\end{align}
The sum on the right side of \eqnref{PalmZ} can be accumulated as rewards over each renewal period. In the $k$th renewal period, we set the  reward $R_k$ equal to the number of slots in the $k$th renewal period in which $Z_n=z$. If $Y_k\ge z$, then there is one slot in the renewal period in which $Z_n=z$ and thus the reward is $R_k=1$; otherwise $R_k=0$.  Thus, for $z=1,2,\ldots$,  $R_k=\I{Y_k\ge z}$. From renewal reward theory, it follows from \eqnref{PalmZ} that
\begin{align}\eqnlabel{Zpmf}
    \prob{Z_n=z}=\frac{\E{R_k}}{\E{Y_k}}=\frac{\prob{Y_k\ge z}}{\E{Y_k}},\quad z=1,2,\ldots.
\end{align}
This is the discrete-time version of a well-known distribution of the age of a renewal process. It follows from \eqnref{Zpmf} that the average age of the current update at the server has average age
\begin{align}
  \E{A_i(n)}=\E{Z_n}&=\sum_{z=1}^\infty z\prob{Z_n=z}=\frac{\E{Y_k^2}}{2\E{Y_k}}+\frac{1}{2}.
  \eqnlabel{EZ-input-age}
\end{align}

\subsection{Sampling the Server: The Age at the Monitor}
Now we examine age at the monitor for the DAD and RAD servers. As depicted in Figure~\ref{fig:DTsampler}, the server (whether DAD or RAD) sends its freshest update to the monitor at sample times $S_1,S_2,\ldots$.  These sample times $S_k$ also form a renewal process with iid inter-sample times $Y'_k=S_k-S_{k-1}$.   For the DAD server, $Y'_k=\tau$. while for the RAD server, the $Y'_k$ are geometric $(\mu)$ random variables.
We now analyze the average age at the monitor in terms of the moments of $Y'_k$. 

The age process $A_i(n)$ at the input to the server is identical to the $Z_n$ process defined in Section~\ref{sec:fresh}. When the server sends its most recent update to the monitor in slot $S_k$, this update has age $A_i(S_k)$ at the start of the slot. The monitor receives the update at the end of the slot with age $A_i(s_k)+1$. Thus, at the start of slot $s_k+1$, the age at the monitor is reset to $A_m(s_k+1)=A_i(s_k)+1$.  Graphically, this is depicted in Figure~\ref{fig:DTsampler}. 

To describe the monitor age $A_m(n)$ in an arbitrary slot $n$, we look backwards in time and define $Z'_n$ as the age of the renewal process defined by the inter-renewal times $Y'_k$ associated with sampling the server.   In Fig.~\ref{fig:DTsampler} for example,  in slot  $n=11$, the last update was sampled by the server at time $s_2=7$ and $Z'_{11}=11-7=4$. 
We note that the $Z'_n$ process is the same as the $Z_n$ process, modulo the inter-renewal times now being labeled $Y'_k$ rather than $Y_k$. In particular the PMF $\prob{Z'_n=z}$ and expected value $\E{Z'_n}$ are described by \eqnref{Zpmf} and \eqnref{EZ-input-age} with $Y_k'$ replacing $Y_k$. 

The age at the monitor in slot $n$  is then 
\begin{align}
    A_m(n)=A_i(n-Z'_n)+Z'_n.
\end{align}
When the age process $A_i(n)$ at the input to the monitor is stationary and the sampling process  that induces $Z'_n$ is independent of the $A_i(n)$ age process, it follows that $A_i(n)$ and $A_i(n-Z'_n)$ are identically distributed. Thus,
the average age at the monitor is
\begin{align}
    \E{A_m(n)}&=\E{A_i(n-Z'_n)}+\E{Z'_n}\nn&=\E{A_i(n)}+\E{Z'_n}.
\end{align}
Employing  \eqnref{EZ-input-age} to evaluate both $\E{A_i(n)}$ and $\E{Z'_n}$, we obtain
\begin{align}
    \E{A_m(n)}= \frac{\E{Y_k^2}}{2\E{Y_k}}+ \frac{1}{2}+ \frac{\E{(Y'_k)^2}}{2\E{Y'_k}} +\frac{1}{2}.
\end{align}

For both the DAD and RAD servers, the source generates packets as a rate $\lambda$ Bernoulli process with 
$\E{Y_k}=1/\lambda$ and $\E{Y_k^2}=(2-\lambda)/\lambda^2$. It follows from \eqnref{EZ-input-age} that $\E{A_i(n)}=1/\lambda$. For the DAD server, $Y_k'=\tau$ is deterministic and \eqnref{EZ-input-age} implies $\E{Z'_n}=(\tau+1)/2$. For the RAD server, the inter-sample times $Y_k'$ at the monitor   are geometric $(\mu)$ and $\E{Z'_n}=1/\mu$. With these observations, we have the following theorem.
\begin{theorem}\thmlabel{DAD-RAD-age}
When the source emits updates as a rate $\lambda$ Bernoulli process, the average age at the monitor is
\begin{align}
    \E{A_m(n)}&=\begin{cases} \displaystyle
    {1}/{\lambda}+(\tau+1)/2, & \text{DAD server,}\\
    {1}/{\lambda}+{1}/{\mu}, & \text{RAD server.}
    \end{cases}
\end{align}
\end{theorem}

\section{Maximal Leakage vs Average AoI}

\begin{figure}[t]
    \centering
    \includegraphics[width=0.5\textwidth]{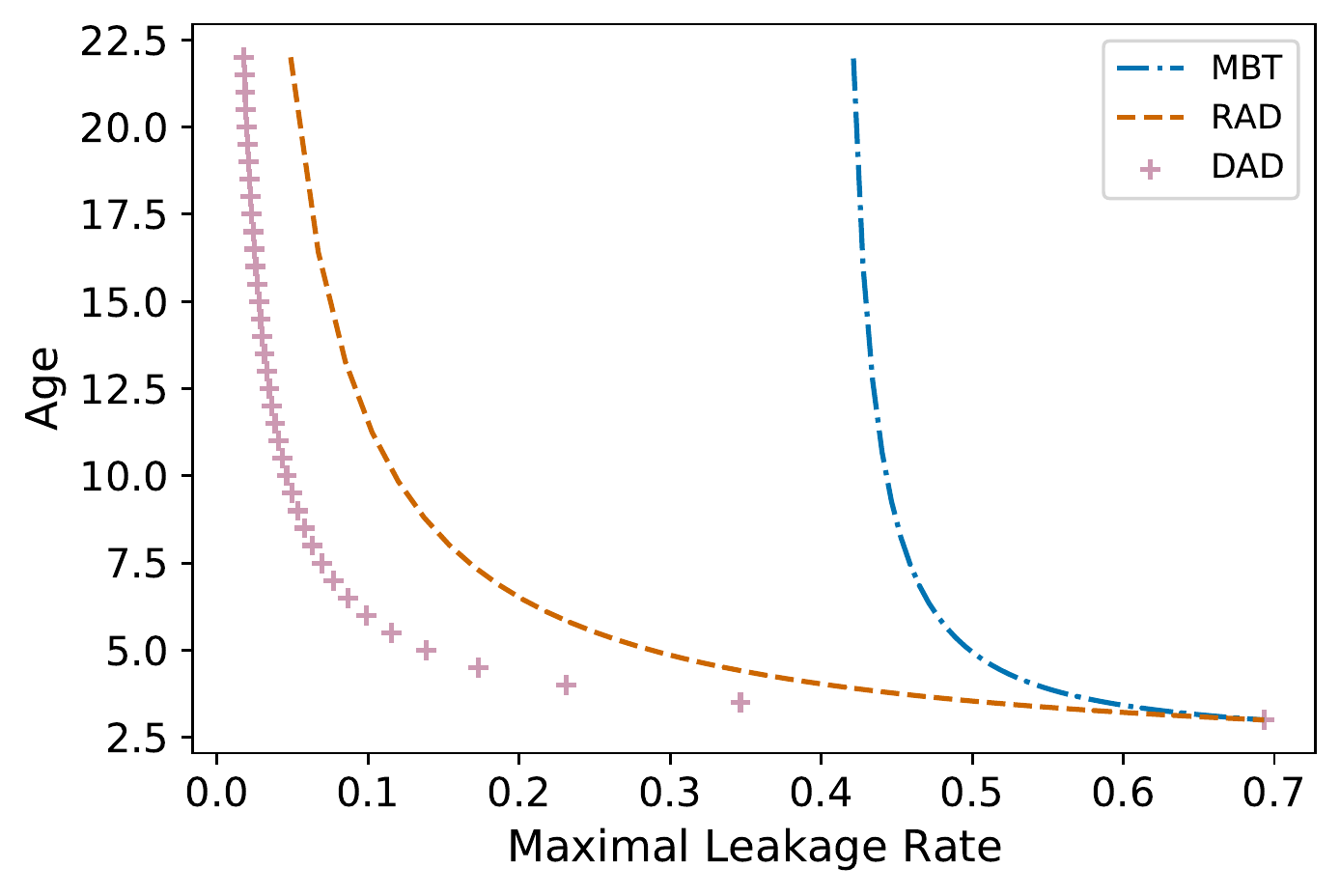}
    \caption{The age vs.~maximal leakage rate  for the MBT $(\alpha=1)$, DAD, and RAD servers with arrival rate $\lambda=0.5$. The service rate $\mu$ varies over $[0.524,1]$ for Geo/Geo/1 and over $[0.05,1]$ for RAD. For DAD, $\tau$ varies from 1 to 39. }
    \label{AvsL}
\end{figure}

Fig.~\ref{AvsL} shows the variation of age with maximal leakage rate for MBT with $\alpha=1$ (which is Geo/Geo/1), DAD, and RAD systems with arrival rate $\lambda=0.5$. As the service rate for these systems increases, the maximal leakage rate for the three service policies increases and the age decreases. Note that for the service rate of $1$ ($\mu=1$ or $\tau=1$), the three systems have identical service service processes and thus the same  average age of 3 slots. 

At other values of the service rate, the DAD system has better age than the other two systems for the same maximal leakage rate. The DAD system keeps the age small by not queueing packets. It delivers only the most recent update packet. From \Thmref{DAD-RAD-age},  we see for the DAD and RAD systems that for a fixed value of the maximal leakage rate (as specified by $\tau$ or $\mu$) the age is decreasing in the source rate $\lambda$ and this is illustrated in  Fig.~\ref{RandomDeterministic}.

\begin{figure}[t]
    \centering
    \includegraphics[width=0.5\textwidth]{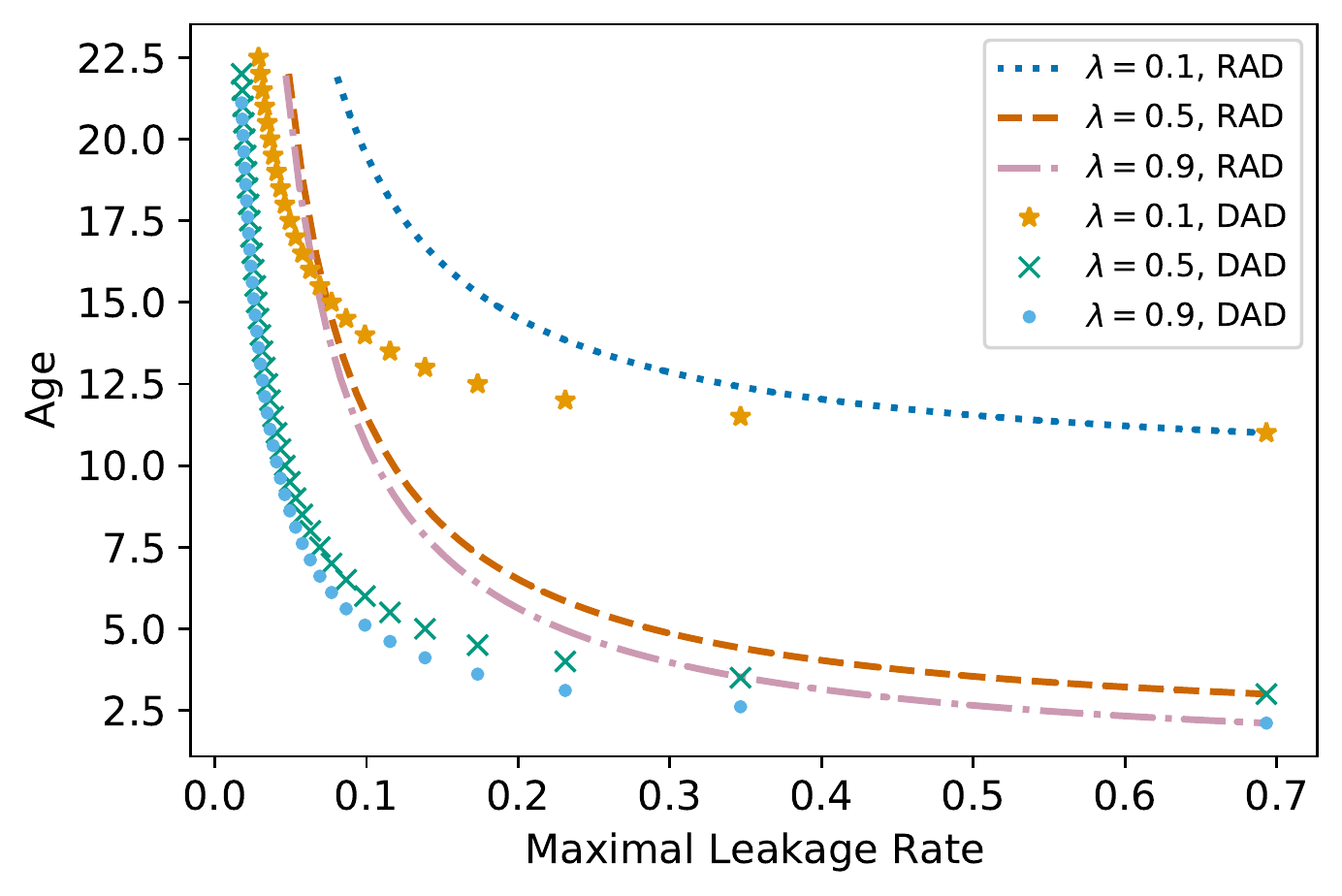}
    \caption{The age vs.~maximal leakage rate for the RAD and DAD policies for $\lambda=0.1$, $\lambda=0.5$ and $\lambda=0.9$. }
    \label{RandomDeterministic}
\end{figure}

In Fig.~\ref{Geo} we show the variation of the age with the maximal leakage rate for the MBT server for various arrival rates $\lambda$. With fixed $\lambda$ and $\alpha=1$, the server controls the leakage via the service rate $\mu$. As $\mu$ is reduced and approaches $\lambda$, the maximal leakage is reduced but the age blows up as 
the queue backlog increases. 
In fact, we can achieve better trade-offs by varying both $\mu$ and the admission probability $\alpha$. 
For a given $\lambda$, we can choose $\alpha$ to operate at an effective arrival rate $\alpha\lambda<\lambda$. This stabilizes the queue by  throwing away arrivals to keep the backlog small. With $\lambda$ fixed, for each $\mu$ (which specifies the leakage), the $\alpha\in(0,1]$ that minimizes $\E{A_{\text{MBT}}}$ in \eqnref{GeoAge} is calculated. This is shown in Fig.~\ref{Geo} by the solid blue-orange-green ``$\alpha\lambda$'' trade-off curve. 
The blue segment is achievable for $\lambda\ge 0.1$, the orange segment is achievable for $\lambda\ge 0.5$ and the green segment is achievable for $\lambda\ge 0.9$.
%

%


\begin{figure}[t]
    \centering
    \includegraphics[width=0.5\textwidth]{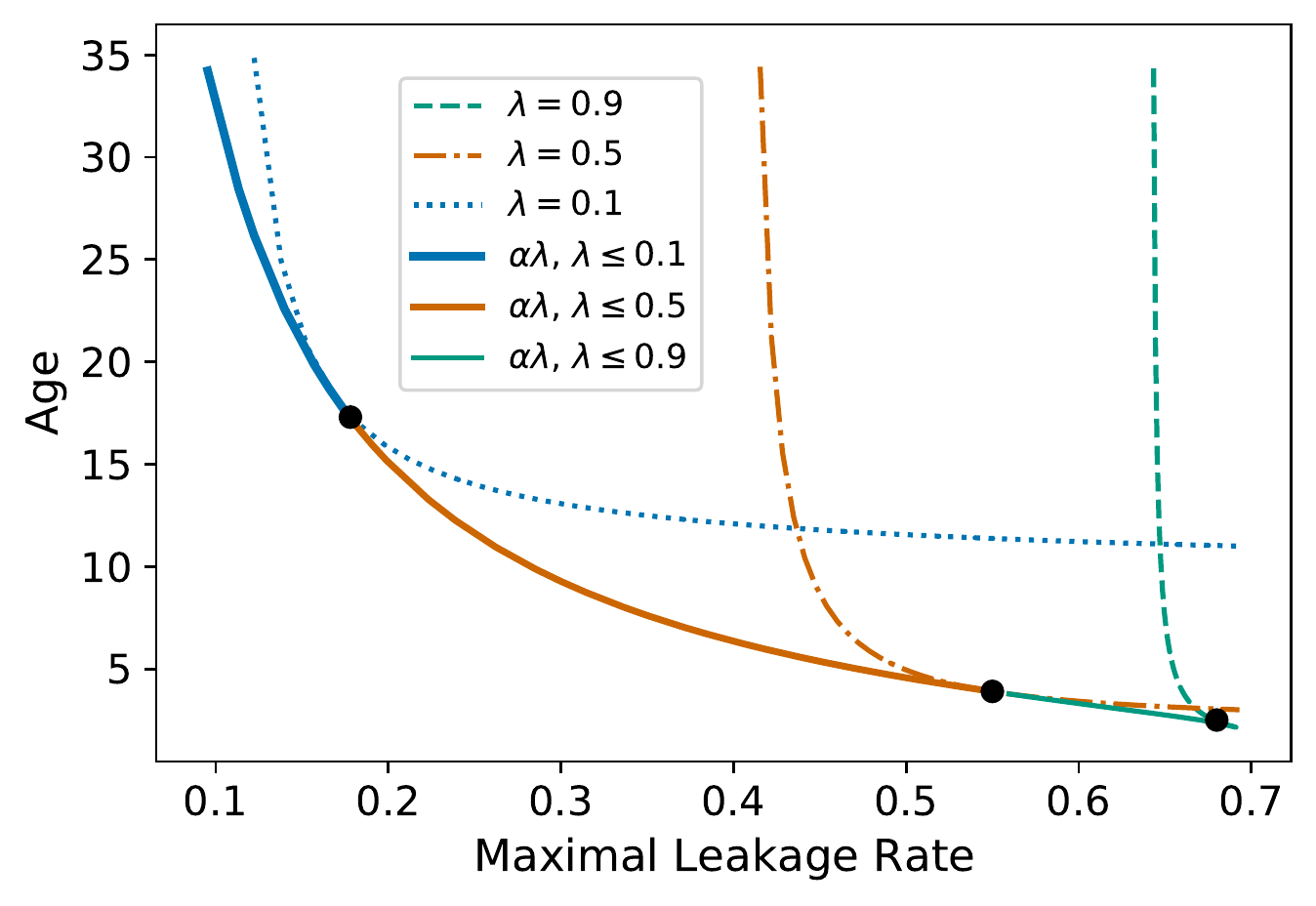}
    \caption{The age vs.~maximal leakage rate for the MBT  system: for $\alpha=1$, the $\lambda=0.1$, $\lambda=0.5$ and $\lambda=0.9$ age-leakage trade-off curves are obtained 
    by varying $\mu$ over the interval $[\lambda+\epsilon,1]$.  With $\lambda$ fixed, for each $\mu$ (which specifies the leakage), the $\alpha\in(0,1]$ that minimizes $\E{A_{\text{MBT}}}$ in \eqnref{GeoAge} is also calculated. This yields the solid blue-orange-green trade-off curve. The blue segment shows $\alpha\lambda\in [0.054,0.1]$ and is achievable for $\lambda\ge 0.1$. The orange segment shows $\alpha\lambda\in [0.1,0.5]$ and the green segment shows $\alpha\lambda\in [0.5,0.9]$; these segments are achievable for $\lambda\ge 0.5$ and $\lambda\ge 0.9$ respectively.}
    \label{Geo}
\end{figure}

\section{Conclusions and Future Work}

In this work we examined the trade-off between the maximal leakage and the average age for three  service policies. 
For a given arrival rate $\lambda$, the DAD server is constrained by integer values of $\tau$ to certain age-leakage operating points. For the same maximal leakage, neither the RAD server nor the MBT server could achieve the same average age. 
However, the MBT and the RAD servers have the flexibility to choose the service rate $\mu$ continuously. In addition, the MBT server can choose to thin the arrival process by operating at an effective arrival rate $\alpha\lambda$ to achieve better age-leakage trade-off.  

Our initial results here suggest several interesting avenues for future investigation. We limited our analysis to Bernoulli arrival processes that simplify the leakage analysis because all input sequences $x^n$ are in the support set. Extending the analysis beyond Bernoulli arrival processes could change which policies have more favorable trade-offs. Other arrival processes may not have support over all possible input sequences $x^n$ which may make the calculation of maximal leakage more challenging. Likewise, we have limited our attention to specific server policies. More generally, we would like to find ``optimal'' policies (within a given class) to manage the age-leakage trade-off for a given arrival process or arrival rate. Finally, understanding these trade-offs in networks of status updating systems might open up new avenues for modeling the interplay between delay and privacy. 




\clearpage

\bibliographystyle{IEEEtran}
\bibliography{paper_draft.bib}
\newpage\clearpage
\section{Appendix}
\subsection*{Proof of Lemma~\ref{max_claim}}
We prove this claim by induction. We start by showing the claim is true for $n=1$.
For $n=1$,
\begin{align}
    \P(Y_1=1|X_1=1)&=\mu, &
    \P(Y_1=1|X_1=0)&=0\eqnlabel{y1givenx11-x10}.
\end{align}
It follows from \eqnref{y1givenx11-x10} that 
\begin{align}
    \max_{x\in \mathcal{X}} P(Y_1=1|X_1=x)=\mu.\eqnlabel{maxy1}
\end{align}
Similarly,
\begin{align}
    \P(Y_1=0|X_1=1)&=1-\mu, &
    \P(Y_1=0|X_1=0)&=1\eqnlabel{y10givenx11-x10}.
\end{align}
It follows from \eqnref{y10givenx11-x10} that 
\begin{align}
    \max_{x\in \mathcal{X}} P(Y_1=0|X_1=x)&=1.\eqnlabel{maxy2}
\end{align}
From \eqnref{maxy1} and \eqnref{maxy2} we can say, 
\begin{align}
    \max_{x\in \mathcal{X}} \pmf{Y_1|X_1}{y_1|x}=\mu^{y_1},
\end{align}
and the maximum is achieved when $x_1=y_1$.
Thus the claim holds for $n=1$.

Now let us say the claim holds for some $n=k$.
Then,
\begin{align}
    \max_{{x^k \in \mathcal{X}}^k} \pmf{Y^k|X^k}{y^k|x^k}=\mu^ {\sum_{i=1}^{k} y_i},
\end{align}
and this maximum is achieved when $x^k=y^k$.

Now we need to show the claim holds for $n=k+1$.
\begin{align}
\spmf{Y^{k+1}|X^{k+1}}{y^{k+1}|x^{k+1}}
&=\spmf{Y^k|X^{k+1}}{y^{k}|x^{k+1}}
\spmf{Y_{k+1}|Y^k,X^{k+1}}{y_{k+1}|y^k,x^{k+1}}.
\end{align}
Since the departure process until time $k$ does not depend on the arrival at time $k+1$, we can write
\begin{align}
    \pmf{Y^k|X^{k+1}}{y^k|x^{k+1}}
    =\pmf{Y^{k}|X^k}{y^{k}|x^k}.
\end{align}
Hence,
\begin{align}
    &\max_{{x^{k+1} \in \mathcal{X}^{k+1}}} \pmf{Y^{k+1}|X^{k+1}}{y^{k+1}|x^{k+1}}\nn
&=\max_{{x^{k+1}} \in {\mathcal{X}}^{k+1}}\spmf{Y^{k}|X^k}{y^{k}|x^k}\spmf{Y_{k+1}|Y^k,X^{k+1}}{y_{k+1}|y^k,x^{k+1}}\nn
&\le [\max_{{\hat{x}^{k+1}} \in {\mathcal{X}}^{k+1}}\spmf{Y^{k}|X^k}{y^{k}|\hat{x}^k}][\max_{{x^{k+1}} \in {\mathcal{X}}^{k+1}}\spmf{Y_{k+1}|Y^k,X^{k+1}}{y_{k+1}|y^k,x^{k+1}}]\nn
&= [\max_{{\hat{x}^{k}} \in {\mathcal{X}}^{k}}\spmf{Y^{k}|X^k}{y^{k}|\hat{x}^k}][\max_{{x^{k+1}} \in {\mathcal{X}}^{k+1}}\spmf{Y_{k+1}|Y^k,X^{k+1}}{y_{k+1}|y^k,x^{k+1}}].
\eqnlabel{twomax}
\end{align}

By the induction hypothesis, 
\begin{align}
    \max_{\hat{x}^{k}\in \mathcal{X}^{k}} \pmf{Y^k|X^k}{y^k|\hat{x}^k}=\mu^ {\sum_{i=1}^{k} y_i},\eqnlabel{max1}
\end{align}
and the maximum is achieved when $\hat{x}^k=y^k$. For the second factor on the right side of \eqnref{twomax}, we consider the cases $Y_{k+1}=1$ and $Y_{k+1}=0$ separately. In both cases, we will employ the binary random sequence $W_k$ such that $W_k=1$ if there is a packet that can be served at time $k$ and otherwise $W_k=0$. We can write $W_{k+1}=g_{k+1}(Y^k,X^{k+1})$,  a deterministic function  of $Y^{k}$ and $X^{k+1}$. This implies
\begin{align}
\eqnlabel{PYkplus1}
&\pmf{Y_{k+1}|Y^k,X^{k+1}}{y_{k+1}|y^k,x^{k+1}}\nn
&=\pmf{Y_{k+1}|Y^k,X^{k+1},W_{k+1}}{y_{k+1}|y^k,x^{k+1},g_{k+1}(y^k,x^{k+1})}.
\end{align}
Moreover, there is a Markov chain \begin{align}
 (Y^k,X^{k+1})\to W_{k+1}\to Y_{k+1}
 \eqnlabel{markov-YX-W-Y1}
\end{align}
since
\begin{subequations}
\eqnlabel{Pygivenw} 
\begin{align}
&\pmf{Y_{k+1}|Y^k,X^{k+1},W_{k+1}}{y_{k+1}|y^k,x^{k+1},0}\nn
&\qquad =
\pmf{Y_{k+1}|W_{k+1}}{y_{k+1}|0}
=
\begin{cases} 1 & y_{k+1}=0,\\
0 & y_{k+1}=1,
\end{cases}\\
&\pmf{Y_{k+1}|Y^k,X^{k+1},W_{k+1}}{y_{k+1}|y^k,x^{k+1},1}\nn
&\qquad=
\pmf{Y_{k+1}|W_{k+1}}{y_{k+1}|1}
=
\begin{cases} 1-p & y_{k+1}=0,\\
p & y_{k+1}=1.
\end{cases}
\end{align}
\end{subequations}
It follows from \eqnref{PYkplus1} and \eqnref{markov-YX-W-Y1} that 
\begin{align}
&\pmf{Y_{k+1}|Y^k,X^{k+1}}{y^{k+1}|y^k,x^{k+1}}\nn
&\qquad=\pmf{Y_{k+1}|W_{k+1}}{y^{k+1}|g_{k+1}(y^k,x^{k+1})}.\eqnlabel{PYgivenWg}
\end{align}
Now consider the maximization of the second factor in \eqnref{twomax}.  For the case $Y_{k+1}=0$, it follows 

\begin{align}
&\max_{x^{k+1}}\pmf{Y_{k+1}|Y^k,X^{k+1}}{0|y^k,x^{k+1}}\nn
&\qquad =\max_{x^{k+1}}\pmf{Y_{k+1}|W_{k+1}}{0|g_{k+1}(y^k,x^{k+1})}\le 1,\eqnlabel{unitybound}
\end{align}
where the unity upper bound 
holds trivially.  However, it follows from \eqnref{Pygivenw} 
that this  unity upper bound is achieved by any $x^{k+1}$ that ensures $W_{k+1}=0$.  Given $Y^k=y^k$, we observe that the upper bound in 
\eqnref{unitybound} is achieved by $(x^k,x_{k+1})=(y^k,0)$. In particular, given the output sequence $y^k$, the input $x^k=y^k$ implies that the system is idle at the end of slot $k$ and, in this case, it follows that  $x_{k+1}=0$ implies $W_{k+1}=0$.
Thus for the case that $Y_{k+1}=0$, we have shown that $\pmf{Y_{k+1}|Y^k,X^{k+1}}{|y^k,x^{k+1}}$ is maximized over $x^{k+1}\in\X^{k+1}$ by $x^{k+1}=(x^k,x_{k+1})=(y^k,0)=y^{k+1}$. 

Now for the case  $Y_{k+1}=1$, it follows from \eqnref{Pygivenw}
\begin{align}
    &\max_{x^{k+1}}\pmf{Y_{k+1}|Y^k,X^{k+1}}{1|y^k,x^{k+1}}\nn
    &=\max_{x^{k+1}}\pmf{Y_{k+1}|W_{k+1}}{1|g_{k+1}(y^k,x^{k+1})}\le p.\eqnlabel{bound2}
\end{align}
The upper bound is achieved by any $x^{k+1}$ that ensures $W_{k+1}=1$. No matter what the past history, $x^{k+1}=1$ ensures $W_{k+1}=1$. It follows that the upper bound in \eqnref{bound2} is achieved by $(x^k,x_{k+1})=(y^k,1)$. Thus for the case that $Y_{k+1}=1$, we have shown that $\pmf{Y_{k+1}|Y^k,X^{k+1}}{1|y^k,x^{k+1}}$ is maximized over $x^{k+1}\in\X^{k+1}$ by $x^{k+1}=(x^k,x_{k+1})=(y^k,1)=y^{k+1}$. 

Now  returning to \eqnref{twomax}, the first factor is maximized by $\hat{x}^k=y^k$ and the second factor is maximized by $x^{k+1}=y^{k+1}$ hence the product is maximized by $x^{k+1}=(\hat{x}^k,x_{k+1})=(y^k,1)=y^{k+1}$. 

Hence,
\begin{align}
    \max_{x^{k+1}\in X^{k+1}} \pmf{Y_{k+1}|Y^k,X^{k+1}}{y^{k+1}|y^k,x^{k+1}}
    &=\mu^ {y_{k+1}}\eqnlabel{max2}
\end{align}

From \eqnref{twomax}, \eqnref{max1} and \eqnref{max2} we get, 
\begin{align}
    \max_{x^{k+1}\in X^{k+1}} \pmf{Y^{k+1}|X^{k+1}}{y^{k+1}|x^{k+1}}
    =\mu^ {\sum_{i=1}^{k+1} y_i},
\end{align}
and the maximum is achieved when $x^{k+1}=y^{k+1}$.

Thus, the claim holds for $n=k+1$, completing the proof of Lemma~\ref{max_claim}.

\end{document}

%% file: macros.tex
\long\def\/*#1*/{}

\renewcommand{\P}{\mathbb{P}}
\renewcommand{\L}{\mathcal{L}}

\newcommand{\X}{\mathcal{X}}
\newcommand{\Y}{\mathcal{Y}}

\usepackage[normalem]{ulem}
\newcommand\redout{\bgroup\markoverwith{\textcolor{red}{\rule[.5ex]{2pt}{0.4pt}}}\ULon}

\newtheorem{theorem}{Theorem}
\newtheorem{definition}{Definition}
\newtheorem{lemma}{Lemma}